\documentclass[11pt]{article}

\usepackage{amssymb}
\usepackage{amsmath}
\usepackage{amsthm}
\usepackage{textcomp}
\usepackage{array}
\usepackage[latin5]{inputenc}
\usepackage{lineno}
\usepackage{color}
\usepackage[margin=1.3in]{geometry}
\usepackage{comment}

\usepackage{graphicx}

\newcommand{\ket}[1]{|#1\rangle}
\newcommand{\braket}[2]{\langle #1|#2\rangle}

\newtheorem{theorem}{Theorem}
\newtheorem{corollary}{Corollary}

\theoremstyle{definition}

\newcommand{\LL}{\mathtt{L}}

\title{Magic coins are useful for small-space quantum machines}
\author{
A. C. Cem Say\\
\small
Bo\u{g}azi\c{c}i University, Department of Computer Engineering, Bebek 34342 \.{I}stanbul, Turkey
\\
Abuzer Yakary{\i}lmaz\thanks{Yakary{\i}lmaz was partially supported by CAPES with grant 88881.030338/2013-01, ERC Advanced Grant MQC, and FP7 FET project QALGO\@. Moreover, the part of the research work was done while Yakary{\i}lmaz was visiting Bo\u{g}azi\c{c}i University in 2014\@.}
\\
\small 
National Laboratory for Scientific Computing, Petr\'{o}polis, RJ, 25651-075, Brazil
\\
\small \texttt{say@boun.edu.tr}~~~~\small \texttt{abuzer@lncc.br}
}

\date{\small Keywords: quantum computation, constant space, unrestricted amplitudes}

\begin{document}

\maketitle

\begin{abstract}
Although polynomial-time probabilistic Turing machines can utilize uncomputable transition probabilities to recognize uncountably many languages with bounded error when allowed to use logarithmic space, it is known that such ``magic coins'' give no additional computational power to constant-space versions of those machines. We show that adding a few quantum bits to the model changes the picture dramatically. For every language $L$, there exists such a two-way quantum finite automaton that recognizes a language of the same Turing degree as $L$ with bounded error in polynomial time. When used as verifiers in public-coin interactive proof systems, such automata can verify membership in all languages with bounded error, outperforming their classical counterparts, which are known to fail for the palindromes language.
\end{abstract}

\section{Introduction}
Determining whether various resource-bounded quantum models are equivalent or superior to their classical probabilistic counterparts in power gives valuable insight about the nature of quantum computation, and about the important cases (notably, polynomial-time models) where this problem is still open. Sometimes, arguably unrealistic setups and resources (e.g. \cite{Aar05,AW09}) are also used in these comparisons, keeping in mind how the manifestly unrealistic concept of nondeterminism forms a priceless tool of complexity theory. In this paper, we will be considering variants  of bounded-error constant-space Turing machines allowed to toss coins with  unrestricted (i.e. possibly uncomputable)  real numbers as the probability of coming up heads. In that regime, it is known  that probabilistic  polynomial-time language recognizers have no additional computational power over their standard, fair-coin versions \cite{DS90}.  Languages that cannot be handled by probabilistic public-coin verifiers, even with such magic coins, have also been demonstrated \cite{DS92}. We show that quantum computers outperform their classical counterparts in these contexts; by demonstrating that every language is Turing-equivalent to some language recognized by a polynomial-time two-way quantum finite automaton, and that every language has a public-coin proof system with such an automaton as the verifier. The amount of ``quantumness'' used in our constructions is very small, as the machines we utilize can be viewed as  deterministic two-way finite automata augmented by just one or   two qubits.

\section{Preliminaries}
It is well known that polynomial-time  probabilistic Turing machines (PTM's) can utilize uncomputable transition probabilities to recognize uncountably many languages with bounded error: The $k$'th bit in the decimal expansion of the probability that a given coin will land heads can be estimated by a procedure that involves tossing that coin for a number of times that is exponential in $k$. Given any language $L$ on the alphabet $\{a\}$, and a coin which lands heads with probability $0.x$, where $x$ is an infinite sequence of digits whose $k$'th member encodes whether the $k$'th unary string is in $L$, the language $\{a^{4^k} \vert a^k \in L\}$ is therefore in $\mathsf{BPP}_\mathbb{R}$, the version of $\mathsf{BPP}$ with arbitrary real transition probabilities.\footnote{We thank Peter Shor for helping us with these facts.} The machine in this construction uses logarithmic space, and it was proven by Dwork and Stockmeyer \cite{DS90} that no polynomial-time PTM using $o(\log \log n)$ space can recognize a nonregular language with bounded error, even with unrestricted transition probabilities.
 
Adleman et al. \cite{ADH97} used a construction similar to the one described above to demonstrate that $\mathsf{BQP}_\mathbb{R}$, the class of languages recognized with bounded error in polynomial time by quantum Turing machines (QTM's) employing arbitrary real amplitudes, contains languages of every Turing degree. The QTM formulated by Adleman et al. also uses logarithmic space. Our construction in Section \ref{sec:rec} shows that the coins described in  \cite{ADH97} can also be exploited by polynomial-time constant-space PTM's augmented by a single qubit.

The constant-space quantum model that we will use is the two-way finite automaton with quantum and classical states (2qcfa), introduced by Ambainis and Watrous \cite{AW02}, in which the quantum and classical memories are nicely separated, allowing a precise quantification of the amount of ``quantumness'' required for the task at hand.\footnote{Our exposition of 2qcfa's will use superoperators, generalizing and simplifying the setup of  \cite{AW02}; see \cite{YS11A}. It is not hard to modify our algorithms for implementation in the format of \cite{AW02}.} These machines can be viewed as two-way deterministic automata augmented with a  quantum register of constant size (typically, just one or two qubits). The input string is assumed to be written between two end-marker symbols on a tape. A 2qcfa starts with both its quantum and classical parts set to be in their respective initial states, and operates separately on these parts at each step:  
\begin{itemize}
	\item First, a superoperator (Figure \ref{fig:superoperators}), determined by the current classical state and the symbol being scanned on the input tape, is applied to the quantum register, yielding an outcome. 
	\item Then, the next classical state and tape head movement direction is determined by the  current classical state, the symbol being scanned on the input tape,  and the observed outcome. 
\end{itemize}

Execution halts when the machine enters an accepting or rejecting classical state.

\begin{figure}[!ht]
	\centering
	\footnotesize
	\fbox{
	\begin{minipage}{0.96\textwidth}
		For a 2qcfa with $j$ quantum states, each superoperator $ \mathcal{E} $ is composed of a finite number of $j\times j$ matrices called \textit{operation elements},
		$ \mathcal{E} = \{ E_{1}, \ldots, E_{k} \} $, satisfying
		\begin{equation}
		\label{eq:completeness}
			\sum_{i=1}^{k} E_{i}^{\dagger} E_{i} = I,
		\end{equation}
		where $ k \in \mathbb{Z}^{+} $, and the indices label the possible outcomes.
		When a superoperator $\mathcal{E}$ is applied to 
		a quantum register in superposition $\ket{\psi}$, 
		one observes the outcome $i$ with probability 
		$ p_{i} = \braket{\widetilde{\psi_{i}}}{\widetilde{\psi_{i}}} $,
		where $\ket{\widetilde{\psi_{i}}}$ 
		is calculated as $ \ket{\widetilde{\psi}_{i}} = E_{i} \ket{\psi} $, and $1 \leq i \leq k$.
		If the outcome $i$ is observed ($p_{i} > 0 $), the new superposition 
		is obtained by normalizing $ \ket{\widetilde{\psi}_{i}} $, 
		yielding $ \ket{\psi_{i}} = \frac{\ket{\widetilde{\psi_{i}}}}{\sqrt{p_{i}}} $.
		Unitary operations such as rotations, as well as measurements and probabilistic branchings can be realized within this framework. 
	\end{minipage}
	}
	\caption{Superoperators (adapted from \cite{YSD14})}
	\label{fig:superoperators}
\end{figure}

Ambainis and Watrous \cite{AW02} demonstrated a 2qcfa with just two quantum states (i.e. a single qubit) and computable transition amplitudes that recognizes the nonregular language  $ \mathtt{EQ} = \{ a^nb^n \mid n \geq 0 \} $, in polynomial expected time, thereby establishing the superiority of such machines over their classical counterparts. In the next section, we will use a modified version of the Ambainis-Watrous algorithm as part of our construction of 2qcfa's of the same Turing degree as any given language.

We will compare the capabilities of probabilistic and quantum finite automata as verifiers in public-coin proof systems in Section \ref{sec:ver}. It is known that $\mathsf{AM_\mathbb{R}(2pfa)}$, the class of languages which have such bounded-error proof systems with two-way probabilistic finite automata utilizing arbitrary real transition probabilities as verifiers, does not include the binary palindromes language \cite{DS92}. 

In  public-coin systems with 2qcfa verifiers \cite{Yak13C}, the outcomes of the superoperators of the verifier are instantly available to the computationally powerful prover, who is trying to convince the verifier that the input should be accepted. The quantum and classical transition functions of the 2qcfa's in these systems are modified to take into account of the communication symbol sent by the prover in each step; see \cite{YSD14} for the details.

Although this paper is the first to study 2qcfa verifiers with unrestricted amplitudes, the class $\mathsf{AM_\mathbb{Q}(2qcfa)}$, i.e. the version with rational amplitudes, has already been shown to include $ \mathsf{PSPACE} $\footnote{The proof will appear in the new version of \cite{Yak13C}.} and some $\mathsf{NEXP}$-complete languages \cite{Yak13C}. The constructions in this papers also involve very small amounts of quantum memory. One technique that we will borrow from \cite{Yak13C,YSD14} is that a 2qcfa verifier can encode the integer represented by a binary string that is read from the prover into the amplitude of a quantum state, albeit with low probability in each such attempt.

Another demonstration of the superiority of quantum finite automata over their classical counterparts in this unrealistic setup with no noise and decoherence is the result of Aaronson and Drucker \cite{AD11}, who showed a succinctness advantage of quantum real-time machines in the task of distinguishing two classical coins with different bias. 

In the following, $ \Sigma^*(i) $ denotes the $ i $th element in the lexicographic ordering of all possible strings on alphabet $ \Sigma $, where $ i> 0 $. $ \Sigma^*(1) $ is the empty string, denoted $\varepsilon$. We will use the fact that, for any $i$, the string $ 1\Sigma^*(i) $ is the binary encoding of $i$. 
For a string $w$, $|w|_\sigma$ denotes the number of occurrences of symbol $\sigma$ in $w$.

\section{Language recognition in polynomial time}
 \label{sec:rec}

In this section, we present a constant-space version of a theorem proven by Adleman et al. \cite{ADH97} for logarithmic-space quantum Turing machines. The main concern in the adaptation  is to get a machine to recognize, and count up to, arbitrary powers of 8, without using any secondary memory. 

\newcommand{\powereq}{\mathtt{POWER\mbox{-}EQ}}
\newcommand{\powereqL}{\mathtt{POWER\mbox{-}EQ(L)}}
\begin{theorem}\label{theorem:powereq}
For every language $L$, there exists a language  $L'$  on the alphabet $\{a,b\}$ such that $L'$ is Turing equivalent to $L$, and there exists a 2qcfa which recognizes $L'$ with bounded error in polynomial expected time.
\end{theorem}
\begin{proof}
We start by considering the language 
\[
	\powereq=\{ a b a^7 b a^{7 \cdot 8} b a^{7 \cdot 8^2} b a^{7\cdot 8^3} b \cdots b a^{7\cdot 8^n} \mid n \geq 0 \},
\]
where the number of $a$'s in  any string  is seen to be $ 8^{n+1} $ for some $n \geq 0  $. The 2qcfa described below (Figure \ref{fig:powereqalg}) uses a single qubit (with quantum state set $\{q_0,q_1\}$) to recognize $ \powereq $ in polynomial time.

\newcommand{\IND}{\hspace{10pt}}

\begin{figure}[!ht]
	\centering
	\footnotesize
	\fbox{
	\begin{minipage}{0.96\textwidth}
If the input is $a b a^7$, ACCEPT.

Check whether the input is of the form $a b a^7 b a^{7 \cdot t_1} b a^{7 \cdot t_2}  \cdots b a^{7 \cdot t_n} $ for some $n> 0 $, where all the $t_i$ are positive multiples of 8. If not, REJECT. 

Move the head to the leftmost $b$ in the input.

LOOP:

\IND 	Move the head right to the next $a$.
 	
\IND	Set the qubit to $ \ket{q_0} $. 

\IND	While the currently scanned symbol is $a$:

\IND \IND	Rotate the qubit with angle $\sqrt{2} \pi $.

\IND \IND	Move the head to the right.

\IND	(The currently scanned symbol is a $b$.) Move the head to the $a$ on the right. 

\IND	While the currently scanned symbol is $a$:

\IND \IND	Rotate the qubit with angle $-\sqrt{2} \pi $.

\IND \IND	Move the head 8 squares to the right.
 			
\IND Measure the qubit. If the result is $ q_1 $, REJECT. (I)
 		
\IND 	If the currently scanned symbol is  a $b$, move the head to the nearest $b$ on the left, and goto LOOP.
 	
\IND 	(The currently scanned symbol is the right end-marker.) 	
 	
\IND 	Repeat twice: 
 	
\IND \IND	Move the tape head to the first input symbol.

\IND \IND	While the currently scanned symbol is not an end-marker, do the following: 

\IND \IND \IND	Simulate a classical coin flip. If the result is heads, move right. Otherwise, move left. 

\IND 	If the process ended at the right end-marker both times, and two more  coin flips both turn out heads, goto EXIT.	

\IND Move the head to the leftmost $b$ in the input, and goto LOOP.

EXIT:

\IND 	ACCEPT the input.
	\end{minipage}
	}
	\caption{A 2qcfa for $\powereq$}
	\label{fig:powereqalg}
\end{figure}

After an easy deterministic check, that 2qcfa enters an infinite loop where each iteration compares $ t_i$ with $\frac{t_{i+1}}{8} $ for all $ i \in {1,\ldots,n-1} $. This is achieved by first rotating the qubit counterclockwise  $ t_i $ times,  then rotating it clockwise $ \frac{t_{i+1}}{8} $ times, and finally checking whether it has returned to its original orientation $ \ket{q_0} $. Since the rotation angle is an irrational multiple of $\pi$, the probability $r_i$ that the machine will reject at the line marked (I) in Figure \ref{fig:powereqalg} is zero if and only if  $ t_i=\frac{t_{i+1}}{8} $ for the corresponding $i$. If $ t_i \neq \frac{t_{i+1}}{8} $, then $r_i$  will be at least \cite{AW02}
\[
	\frac{1}{2 (t_i - \frac{t_{i+1}}{8} )^2 }>\frac{1}{2 (t_i+ \frac{t_{i+1}}{8} )^2 }.
\]
We therefore conclude that any input string $ w \notin \powereq$ which has survived the deterministic check in the beginning will be rejected with a probability greater than $ \frac{1}{2 |w|^2} $ in each iteration of the infinite loop.

If the input $w$ has not been rejected after all the $n-1$ comparisons described above, the 2qcfa makes two consecutive random walks starting on the first input symbol, and ending at an end-marker. The probability that both these walks will end at the right end-marker, and the subsequent coin tosses both yield heads, leading to acceptance in this iteration of the infinite loop, is $ \frac{1}{4(|w|+1)^2} $, and the expected runtime for this stage is $ O(|w|^2) $ \cite{AW02}. This means that the machine will halt within $O(|w|^2)$ expected iterations of the loop, leading to an overall expected runtime of $ O(|w|^4) $. 

To conclude, the 2qcfa of Figure \ref{fig:powereqalg} will accept any string $w\in\powereq$ with probability 1. On the other hand, any $w\notin\powereq$ that makes it into the loop has a rejection probability that is more than twice as large as its  acceptance probability in each iteration, and therefore will be rejected with probability greater than $ \frac{2}{3} $. 

We are now ready to adapt the technique of Adleman et al. \cite{ADH97} to 2qcfa's, since  the algorithm in Figure \ref{fig:powereqalg} provides us with a way of ensuring that the number of $a$'s on the tape is a power of 8  with sufficiently high reliability.

For any  $ \mathtt{L} \subseteq \Sigma^* $, define the  language 
\[
	\powereqL = \{ w \in \{a,b\}^* \mid w \in \powereq \mbox{ and } \Sigma^*(\log_8 (|w|_a)) \in \mathtt{L}  \}.
\]
Note that $\mathtt{L}$ and $\powereqL$ are Turing reducible to each other.

We will show that, for any $\mathtt{L}$, $\powereqL $ can be recognized by the 2qcfa described in Figure \ref{fig:powereqLalg} in polynomial expected time:

\begin{figure}[!ht]
	\centering
	\footnotesize
	\fbox{
	\begin{minipage}{0.80\textwidth}
(Assume that the input string is in $\powereq$.)

Move the head to the left end-marker.

Set the qubit to $ \ket{q_0} $. 

While the currently scanned symbol is not the right end-marker:

\IND	Rotate the qubit with angle $\theta_\mathtt{L} $ only if the currently scanned symbol is an $a$.~

\IND 	Move the head to the next square on the right.

Rotate the qubit with angle $\frac{\pi}{4} $.

Measure the qubit. If the result is $ q_1 $, ACCEPT. Otherwise, REJECT.

	\end{minipage}
	}
	\caption{A 2qcfa for $\powereqL$}
	\label{fig:powereqLalg}
\end{figure}

Any string which is a member of $\powereq$ contains $8^{j}$ $a$'s for some $j$, and that string is in $\powereqL$ if $\Sigma^*(j)\in \mathtt{L}$. This information about $\mathtt{L}$ is encoded into the transition amplitudes via the angle $\theta_{\mathtt{L}}$ as follows.

Define 
\[
	\theta_{\mathtt{L}} = 2 \pi \sum_{i=1}^{\infty} \left( \frac{F_{\mathtt{L}}(i)}{8^{i+1}}  \right) 
\]
where the function $ F_{\mathtt{L}}:\mathbb{N}\rightarrow \{-1,1\} $ is
\[
	F_{\mathtt{L}}(n) = \left\lbrace \begin{array}{rl}
		1, & \mbox{ if } \Sigma^*(n) \in \mathtt{L},
		\\
		-1, & \mbox{ if  } \Sigma^*(n) \notin \mathtt{L}.
	\end{array}
	\right. 
\]

Rotating the qubit $8^{j}$ times with $\theta_{\mathtt{L}}$ radians
 leaves it at an angle
\[
	8^{j} \cdot  2 \pi \sum_{i=1}^{\infty} \left( \frac{F_{\mathtt{L}}(i)}{8^{i+1}}  \right) =
	 \pi \frac{F_{\mathtt{L}}(j)}{4} +  \left( 2\pi \frac{F_{\mathtt{L}}(j+1)}{8^2} + 2\pi \frac{F_{\mathtt{L}}(j+2)}{8^3}  +\cdots \right)~~(\mathrm{mod}~2\pi)
\]
radians from the original orientation $\ket{q_0}$. After the last rotation by $\frac{\pi}{4}$ radians, the qubit's final angle from $\ket{q_0}$ is  $ \frac{\pi}{2} + \delta $ (i.e. near $\ket{q_1}$) if $ F_{\mathtt{L}}(j) = 1 $, and $ \delta $ (i.e. near $\ket{q_0}$) if $F_{\mathtt{L}}(j)=-1$ for a $ \delta $ guaranteed to be sufficiently small to 
obtain an error bound of 0.02, as in \cite{ADH97}. The runtime is linear in the input size.

All that remains is to combine the algorithms of Figures \ref{fig:powereqalg} and \ref{fig:powereqLalg}. One starts by using the technique depicted in Figure \ref{fig:powereqalg} to reject strings that are not in $\powereq$ with probability at least 0.666. Members of $\powereq$ are treated correctly by the 2qcfa of Figure \ref{fig:powereqLalg} with probability 0.98. We conclude that the combined 2qcfa recognizes $\powereqL$ with probability at least 0.65, and this probability can be reduced further as desired using standard repetition techniques.

The expected running time of the combined algorithm is mainly 
due to the method of Figure \ref{fig:powereqalg}, and is again polynomially bounded. 
\end{proof}

Some automata-theoretic corollaries of this result are presented in the Appendix. 


\section{Public-coin proof systems}
\label{sec:ver}

In our demonstration of the power provided to verifiers by quantum coins with uncomputable bias, we will use a slightly different method than the one in Section \ref{sec:rec} to represent an entire language by a real number.
 


For any given language $\mathtt{L} $ on alphabet $\Sigma$, let the characteristic function $ G_\mathtt{L} $ map any member of $\mathtt{L}$ to 1, and any non-member to  0. The number 
\[
	\gamma_\mathtt{L} = \sum_{i=1}^\infty \frac{ G_\mathtt{L}( \Sigma^*(i) )}{4^i} = \frac{ G_\mathtt{L}( \Sigma^*(1) ) }{4} +\frac{ G_\mathtt{L}( \Sigma^*(2) ) }{4^2} + \frac{ G_\mathtt{L}( \Sigma^*(3) ) }{4^3} + \cdots
\]
will be used to encode information about $\mathtt{L}$ into the bias of a quantum coin, as will be explained shortly.
Note that $ \gamma_\mathtt{L}$ equals 0 for $\mathtt{L}= \{ \}$, and $ \frac{1}{3} $ for  $\mathtt{L}= \Sigma^* $. For any other language $ \mathtt{L} $, $ \gamma_\mathtt{L} $ is a real number between 0 and $ \frac{1}{3} $.

We start with the case of tally languages, which allows a simpler illustration of the main technique. 

\begin{theorem}\label{theorem:unary}
	For any language $ \mathtt{L} $ on the unary alphabet $U=\{a\} $, there exists a bounded-error public-coin interactive proof system where the messages are classical, the verifier is a 2qcfa with a single quantum bit, and the expected runtime for inputs of length $n$ is $2^{O(n)}$.
\end{theorem}
\begin{proof}
We describe the proof system in question. For any input of length $n\geq 0$, the prover is supposed to send a stream of bits describing the membership status of every string of length at most $n$, i.e. the sequence 
\[
	u_{\LL,n} =  G_\LL(\varepsilon)  G_\LL(a)  G_\LL(aa)  G_\LL(aaa)\cdots G_\LL(a^k)
\] 
to the verifier. This transmission is controlled by the outcomes of the coins of the verifier, which are of course visible to the prover. As will be seen below, some outcomes will be interpreted as commands for interrupting the transmission and restarting it from the first bit, whereas others will simply mean ``go on with the next bit''.

Figure \ref{fig:verifiercoins} depicts the  superoperators corresponding to the ``coins'' of the verifier. The outcome associated with each operation element appears as the subscript in its name. The reader may check that all four superoperators obey the wellformedness condition (\ref{eq:completeness}) stated in Figure \ref{fig:superoperators}. The initial state of the quantum part is the one corresponding to the first row and column in the operation elements, so the initial superposition of the qubit is just $\left( \begin{array}{c}
			1 \\ 0
\end{array}		 \right)$.

\begin{figure}[!ht]
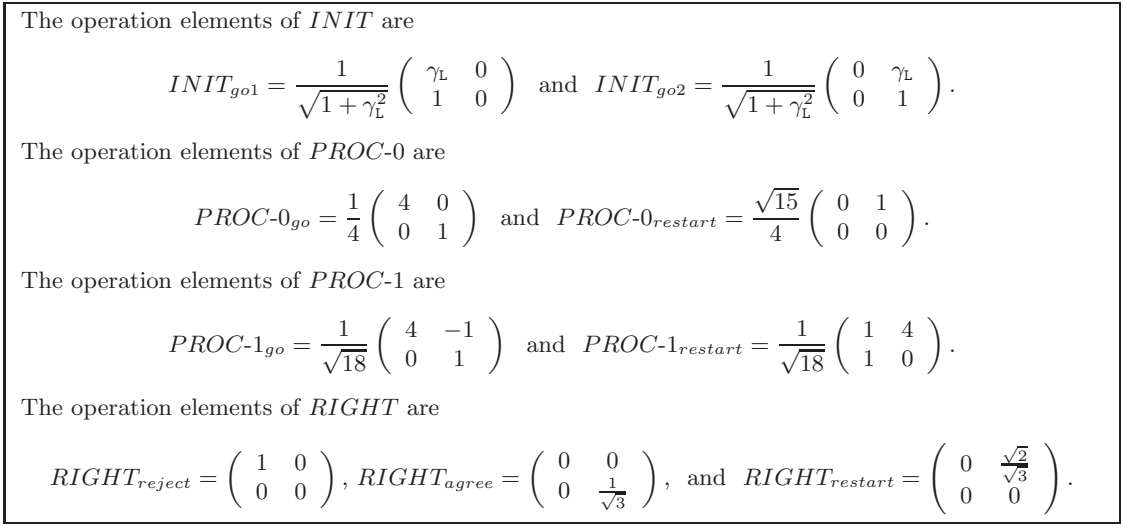

	\centering
	\footnotesize
	\fbox{
	\begin{minipage}{0.96\textwidth}
		The operation elements of 
			$ INIT $  are
			\[
				INIT_{go1} = \frac{1}{\sqrt{1+\gamma_{\LL}^2}} \left( 
					\begin{array}{cccr}
						\gamma_{\LL} & 0  \\ 
						1 & 0 
					\end{array} 
				\right)
				\mbox{~~and~~}
				INIT_{go2} =  \frac{1}{\sqrt{1+\gamma_{\LL}^2}} \left( 
					\begin{array}{cccr}
						0 & \gamma_{\LL}   \\ 
						0 & 1 
					\end{array} 
				\right).
			\]
The operation elements of 
			$ PROC\mbox{-}0 $  are
			\[
				PROC\mbox{-}0_{go} = \frac{1}{4} \left( 
					\begin{array}{cccr}
						4 & 0  \\ 
						0 & 1 
					\end{array} 
				\right)
				\mbox{~~and~~}
				PROC\mbox{-}0_{restart} = \frac{\sqrt{15}}{4} \left( 
				\begin{array}{cccr}
						0 & 1  \\ 
						0 & 0
				\end{array} 
				\right).
			\]
The operation elements of 
			$ PROC\mbox{-}1 $  are
			\[
				PROC\mbox{-}1_{go} = \frac{1}{\sqrt{18}} \left( 
					\begin{array}{cccr}
						4 & -1  \\ 
						0 & 1 
					\end{array} 
				\right)
				\mbox{~~and~~}
				PROC\mbox{-}1_{restart} = \frac{1}{\sqrt{18}} \left( 
				\begin{array}{cccr}
						1 & 4  \\ 
						1 & 0 
				\end{array} 
				\right).
			\]
The operation elements of 
			$ RIGHT $  are
			\[
				 RIGHT_{reject} = \left( 
					\begin{array}{cccr}
						1 & 0  \\ 
						0 & 0 
					\end{array} 
				\right)
				\mbox{, }
				 RIGHT_{agree} = \left( 
					\begin{array}{cccr}
						0 & 0  \\ 
						0 & \frac{1}{\sqrt{3}} 
					\end{array}
				\right)
				\mbox{,~~and~~}
				 RIGHT_{restart} = \left( 
					\begin{array}{cccr}
						0 & {\frac{\sqrt{2}}{\sqrt{3}}}   \\ 
						0 & 0
					\end{array}
				\right).
			\]
\end{minipage}
	}
	\caption{Nontrivial superoperators used by the verifier in the proof of Theorem \ref{theorem:unary}}
	\label{fig:verifiercoins}
\end{figure}

The first action taken by the verifier is an application of the superoperator $INIT$ to the qubit, setting it to the superposition 
	\begin{equation}
	\label{init}
		\frac{1}{\sqrt{1+\gamma_{\LL}^2}} \left( \begin{array}{c}
			\gamma_{\LL} \\ 1
\end{array}		 \right).
	\end{equation}
The outcome ``\textit{go}'' associated with this application cues the prover to send the first bit of  $ u_{\LL,n} $, i.e. the membership indicator of the empty string in $\LL$. The input head of the verifier is on the left end-marker at this point. 

For each bit $b\in \{0,1\}$ sent by the prover, the verifier applies the corresponding superoperator $PROC\mbox{-}b$ to the qubit. If the outcome turns out to be a  ``\textit{go}'', the verifier moves its input head one square to the right, and the same process is repeated unless the right end-marker is scanned. Any ``\textit{restart}'' outcome observed at any point causes the verifier to move the  head back to the left end-marker and restart the program.

Assume that the first $j$ symbols of the input have been processed in this manner without any restarts occurring. Write the qubit's superposition at that point in the form
	\begin{equation}\label{eq:vector}
	\alpha \left( \begin{array}{c}
			\delta \\ 1
\end{array}		 \right),
	\end{equation}
for two real numbers $\alpha$ and $\delta$, where  $\alpha=\frac{1}{\sqrt{1+\delta^2}}$. For $j=0$,  $\delta=\gamma_{\LL}$, as stated above.
If the prover now sends a 0 (respectively, a 1), and no restart occurs due to the application  of the corresponding superoperator, the next superposition would be
	\[
		\beta \left( \begin{array}{c}
			4\delta \\ 1
\end{array}		 \right)
				\mbox{(respectively, }
		\eta \left( \begin{array}{c}
			4\delta-1 \\ 1
\end{array}		 \right),)
\]
for some  reals $\beta$ and $\eta$.

Recalling that $\gamma_{\LL}\in [0,1/3]$, it is important to note that,
	as long as the prover sends membership bits consistent with the encoding of $\LL$ in $\gamma_{\LL}$, the value of $\delta$ in Expression (\ref{eq:vector}) will always be in the interval $[0,1/3]$ for any $j\leq n$.
	Otherwise, this value jumps out of the interval $(-2/3,1) $ at the point of the first disagreement between the prover transmission and $\gamma_{\LL}$, and never returns to this interval later in the execution for greater values of $j$, no matter what the prover says.
	
		If the head reaches the right end-marker, 
	the verifier applies the superoperator $RIGHT$, described in Figure \ref{fig:verifiercoins} when scanning the right end-marker. The three sides of this ``coin'' lead to rejection, or a decision to ``agree'' with the prover (i.e. give the same decision as the prover's claim about the string $a^k$ communicated in the last transmitted bit), or  a restart, as indicated in the figure.
	
As discussed in \cite{YS10B}, the overall acceptance probability of such a ``program with restart'' equals the ratio of the probability of acceptance in a single ``round'' without any restarts to the total probability of halting in such a round. This simplifies the analysis of the behavior of our verifier.

	Prior to the application of the superoperator $RIGHT$, the qubit will again be in a superposition of the form in Expression (\ref{eq:vector}), say,
\[
\alpha_{final} \left( \begin{array}{c}
			\delta_{final} \\ 1
\end{array}		 \right).
\]
 The automatic rejection probability due to the operation element $RIGHT_{reject}$ is then $\alpha_{final}^2\delta_{final}^2$, whereas the probability of adopting the prover's decision due to $RIGHT_{agree}$ is $\alpha_{final}^2/3$.	
	
	 If the input string is a member of $\LL$, then in any particular round, a prover obeying the protocol will cause the verifier to accept with probability $\alpha_{final}^2/3$, whereas the rejection probability  will be at most $\alpha_{final}^2/9$, since $\delta\leq 1/3$ in that case, as explained above. The overall acceptance probability is therefore at least 3/4. 
	
	If the input $a^n$ is not in $\LL$, then the verifier has to transmit an incorrect value for at least the bit about $a^n$ to avoid a rejection probability of 1. But in this case, the absolute value of the amplitude of the first quantum state will be at least $2\alpha_{final}/3$ as explained above, leading to a rejection probability of at least $4\alpha_{final}^2 /9$ in a single round. The overall  probability of rejection would then be at least 4/7. 
	
To provide a bound on the expected runtime of this algorithm, we examine the probability that the  procedure will halt in a single round. Note that 
both operation elements of the  $INIT$ operator leave the qubit in superposition in Expression [\ref{init}], so the analysis for the first round is the same as the ones which are caused by restarts. All the remaining operators will find the qubit in the form of Expression (\ref{eq:vector}). Scrutiny of Figure \ref{fig:verifiercoins} shows that an application of $PROC\mbox{-}0$ will yield a ``go'' with probability at least 1/16, and $PROC\mbox{-}1$ will yield a ``go'' with more than 0.05 probability at each application. The $RIGHT$ operator ends up with a decision to halt with at least 1/3 probability.\footnote{These bounds are calculated by considering the value of $\delta$ that maximizes the restarting probability in each case.} This leads to an upper bound of $O(20^n)$ for the expected number of rounds. Since each run takes linear time, we conclude that the verifier runs in time $2^{O(n)}$.
\end{proof}

\begin{theorem}\label{theorem:binary}
	For any language $ \mathtt{L} $ on the binary alphabet $B=\{0,1\} $, there exists a bounded-error public-coin interactive proof system where the messages are classical, the verifier is a 2qcfa with two quantum bits, and the expected runtime for inputs of length $n$ is $2^{2^{O(n)}}$.
\end{theorem}

\begin{proof}
We construct the proof system in question. For an input string $w$, the transmission expected from a truthful prover will be
\[
	b_{\LL,w} = \# \varepsilon \ddagger G_\LL(\varepsilon) \# 0 \ddagger G_\LL(0) \# 1 \ddagger G_\LL(1) \# 00 \ddagger G_\LL(00) \# 01 \ddagger G_\LL(01) \# 10 \ddagger G_\LL(10) \cdots \# w \ddagger G_\LL(w),
\]
containing the strings and their membership bits in lexicographic order up to $w$. (This transmission can be interrupted and restarted by verifier coin outcomes, as we saw in the previous proof.) Note that the first two symbols of $ b_L $ are $ \ddagger \ddagger $, since $\varepsilon$ has length 0. 
 
The verifier has four quantum states, to which we will refer with their places in the state vector. It sets its quantum register to the superposition \[
		\frac{1}{\sqrt{\gamma_{\LL}^2+3}}\left( \begin{array}{c}
			\gamma_{\LL} \\ 1 \\ 1 \\1
\end{array}		 \right)
	\]
	at the beginning of the computation, and after every restart, as explained in the proof of Theorem \ref{theorem:unary}. The prover then begins the transmission of what is supposed to be $b_{\LL,w}$, bit by bit. Figure \ref{fig:veralgo} describes the action of the verifier on 
 each block of the form $\# s \ddagger \sigma $ of the prover message, where $ s \in B^* $, and $\sigma \in B$.  Note that the procedure depicted in Figure \ref{fig:veralgo} can be interrupted by a ``restart'' outcome, as in the proof of Theorem \ref{theorem:unary}. 
 
\begin{figure}[!ht]
	\centering
	\footnotesize
	\fbox{
	\begin{minipage}{0.96\textwidth}
\begin{itemize}
	\item Perform the following two tasks in parallel while reading the prover's description of the  string $s$:
 		\begin{itemize}
			\item Scan the input string $w$ from left to right, comparing it to $s$. If $s$ is longer than $w$,  reject. 
			\item Encode the binary integer $ 1s $ into the amplitude of the third quantum state.
		\end{itemize}
	\item When the $\ddagger$ symbol, indicating the end of $s$, is received, 
compare the amplitude third quantum state with that of the fourth one,
rejecting with some probability only if they are different. 
\item Treat the membership bit $\sigma$ in the same manner as  in Figure \ref{fig:verifiercoins}, computing a new value for the amplitude of the first state. If $s=w$, the decision on the input 
will be given after the processing of $\sigma$ at the end of this block, in the same way as in Theorem \ref{theorem:unary}, by executing the superoperator $DECIDE$ described in Figure \ref{fig:encoding}.
\end{itemize}
\end{minipage}
	}
	\caption{Verifier procedure on a segment $\# s \ddagger \sigma $ of the prover message}
	\label{fig:veralgo} 
\end{figure}

This procedure  is implemented by the application of the superoperators described in Figure \ref{fig:encoding}. The operation elements corresponding to restarts, which have been omitted from the figure in the interest of brevity, and the common coefficient $c$, can be constructed easily to satisfy Equation \ref{eq:completeness}.  It can be seen that, if one starts in a superposition of the form 
	\[
		\alpha\left( \begin{array}{c}
			\beta \\ 1 \\ 1 \\ K
\end{array}		 \right),
	\]  
	then the application of a sequence of operation elements $ENCODE\mbox{-}b_{go}$ corresponding to the ordering of bits $b$ in $s$ leave the register in a superposition of the form
		\begin{equation}\label{nk}
				\iota\left( \begin{array}{c}
			\beta \\ 1 \\ N \\ K
\end{array}		 \right),
		\end{equation}
where $\alpha$, $\beta$, $\zeta$ and $K$ are some real numbers, and $N$ is the integer represented by the binary string $1s$.) 

The superoperator $SUCC$ (Figure \ref{fig:encoding}), applied when the input symbol $\ddagger$ is scanned, has the effect of rejecting the input with some probability that is nonzero only if $N\neq K$, restarting with some of the remaining probability, and going on with $N+1$ encoded into the fourth amplitude for being used later in the processing of the next segment otherwise. 

The superoperators   $PROC\mbox{-0}$ and $PROC\mbox{-1}$ for treating the membership bit are completely analogous to their namesakes in Figure \ref{fig:verifiercoins}.

The $DECIDE$ superoperator is applied only at the end of the final transmission block of each round, where the prover sends the actual input string and its purported membership bit.

\begin{figure}[!ht]
	\centering
	\footnotesize
	\fbox{
	\begin{minipage}{0.96\textwidth}
	\[
	\begin{array}{llllll}
			ENCODE\mbox{-}0_{go}
			&
			=
			&
			c \left( 
			\begin{array}{cccr}
				1 & 0 & 0 & 0 \\ 
				0 & 1 & 0 & 0 \\ 
				0 & 0 & 2 & 0 \\ 
				0 & 0 & 0 & 1
			\end{array} 
		\right)
		&
		ENCODE\mbox{-}1_{go} 
		&
		=
		&
		c \left( 
			\begin{array}{cccr}
				1 & 0 & 0 & 0 \\ 
				0 & 1 & 0 & 0 \\ 
				0 & 1 & 2 & 0 \\ 
				0 & 0 & 0 & 1
			\end{array} 
		\right)
		\\
		\\
		SUCC_{go}
		&
		=
		&
		c \left( 
			\begin{array}{cccr}
				1 & 0 & 0 & 0 \\ 
				0 & 1 & 0 & 0 \\ 
				0 & 1 & 0 & 0 \\ 
				0 & 1 & 1 & 0
			\end{array} 
		\right)
		&		
		SUCC_{reject}
		&
		=
		& 
		c \left( 
			\begin{array}{cccr}
				0 & 0 & 0 & 0 \\ 
				0 & 0 & 0 & 0 \\ 
				0 & 0 & 1 & -1 \\ 
				0 & 0 & 0 & 0
			\end{array} 
		\right)
	\\
	\\
	PROC\mbox{-}0_{go}
	&
	=
	&
		c \left( 
			\begin{array}{cccr}
				4 & 0 & 0 & 0 \\ 
				0 & 1 & 0 & 0 \\ 
				0 & 0 & 1 & 0 \\ 
				0 & 0 & 0 & 1
			\end{array} 
		\right)
		&
		PROC\mbox{-}1_{go} 
		&
		=
		&
		c \left( 
			\begin{array}{cccr}
				4 & -1 & 0 & 0 \\ 
				0 & 1 & 0 & 0 \\ 
				0 & 0 & 1 & 0 \\ 
				0 & 0 & 0 & 1
			\end{array} 
		\right)
	\\
	\\
	DECIDE_{reject}
	&
	=
	&
		c \left( 
			\begin{array}{cccr}
				1 & 0 & 0 & 0 \\ 
				0 & 0 & 0 & 0 \\ 
				0 & 0 & 0 & 0 \\ 
				0 & 0 & 0 & 0
			\end{array} 
		\right)
		&		
		DECIDE_{agree}
		&
		=
		&
		 c \left( 
			\begin{array}{cccr}
				0 & 0 & 0 & 0 \\ 
				0 & \frac{1}{\sqrt{3}} & 0 & 0 \\ 
				0 & 0 & 0 & 0 \\ 
				0 & 0 & 0 & 0
			\end{array} 
		\right)
	\end{array}
	\]
\end{minipage}
	}
	\caption{Nontrivial operation elements used by the verifier in the proof of Theorem \ref{theorem:binary}}
	\label{fig:encoding} 
\end{figure}

It can be seen that the basic idea is the same as in Theorem \ref{theorem:unary}, with the added complication that an evil prover can now trick the finite-state verifier about which string it is reporting a membership bit for. The solution is to have the prover spell out the strings, and the verifier to compare the encoding of each string with that of the previously sent one to try to catch any such tricks.

Let us analyze this algorithm by tracing the unconditional probabilities of acceptance and rejection in a single round: If the input string is in $\LL$,  the $SUCC$ operator will never cause a rejection. The probability that the input is rejected is at most $(c^j \frac{1}{3})^2$, and the acceptance probability is $(c^j \frac{1}{\sqrt{3}})^2$, for an integer $j$ dependent on the input, leading to overall acceptance with probability $\frac{3}{4}$, as in the proof of Theorem \ref{theorem:unary}.

What if the input $w$ is not in $\LL$? If the prover sends the strings up to $w$ in correct lexicographic order and attempts to sneak in an incorrect membership bit, the verifier rejects with an overall probability of $\frac{4}{7}$, as analyzed in the proof of Theorem \ref{theorem:unary}. If a string is indeed presented out of order, the $SUCC$ operator will catch this and reject with probability  $p=(c^j (N-K))^2$ (cf. Expression (\ref{nk})) for some $j>0$. Any acceptance decision that may occur later as a result of the $DECIDE_{agree}$ operation element would necessarily have probability $(c^k \frac{1}{\sqrt{3}})^2$ for some $k>j$, and it is easy to see that this is small in comparison with $p$, leading to the conclusion that $w$ will be rejected with high overall probability.

Since $b_{\LL,w} $ has length $ 2^{O(n)} $, the halting probability in each round is $c^{-2^{O(n)}}$, yielding an expected running time double exponential in $ n $.
\end{proof}

For any decidable language $ \LL $, $ \gamma_\LL $ is a computable number. Proof systems like the ones described in Theorems \ref{theorem:unary} and \ref{theorem:binary}, but with only computable amplitudes, therefore exist for all decidable languages. 

\section{Open questions}
The capabilities of bounded error PTM's utilizing $
o(\log n)$ space and possibly uncomputable transition probabilities, both as recognizers and verifiers, are unknown. Can they handle uncountably many languages, like their quantum counterparts?

Debate systems \cite{YSD14} are generalizations of proof systems, where an additional agent, the refuter, is trying to convince the verifier that the prover is wrong. Are the classical or quantum versions of these systems able to make better use of uncomputable transitions than the single-prover systems?

Wang \cite{Wan92} showed that the  class  of  all  languages  recognized by 2pfa's using rational transition probabilities is  strictly  contained  in  the  class  of  all  deterministic  context-sensitive  languages. Can 2pfa's with uncomputable transitions recognize any undecidable language with bounded-error? Does the answer change if we allow the machine to use sublogarithmic space (or a counter, or pebbles, or multiple heads, etc.)?

What can be said in this regard about 2pfa's with uncomputable transitions that are verifiers in private-coin proof systems?
\section*{Acknowledgement}
We thank Peter Shor for his helpful answers to our questions.

\bibliographystyle{plain}
\bibliography{tcs}

\appendix

\section{Some corollaries of Theorem \ref{theorem:powereq}}

A language is said to be \textit{stochastic} if there exists a one-way probabilistic finite automaton (1pfa) that accepts any member with probability greater than $ \frac{1}{2} $, and rejects any non-member with probability at most $ \frac{1}{2} $. In his seminal paper, Rabin proved that the cardinality of the class of stochastic languages is uncountable \cite{Rab63}.\footnote{Rabin's proof uses a binary alphabet. Recently, it was shown that \cite{SY14A} the cardinality of stochastic languages on a unary alphabet is also uncountable.} The constructions in Section \ref{sec:rec} allow us to give an alternative proof of this result.

\begin{theorem}
	There exist uncountably many stochastic languages.
\end{theorem}
\begin{proof}	
	We will prove that $ \powereqL $ is a stochastic language for any given language $ \mathtt{L} $. This will complete the proof, since there are uncountably many $ \mathtt{L} $, and a different $ \powereqL $  for any  $ \mathtt{L} $.
	
It is known that the class of languages recognized by the one-way versions of 2qcfa's (namely, the 1qcfa's \cite{ZQLG12}) or any other one-way qfa variant with cut-point $\frac{1}{2}$, i.e. as described in the definition of stochastic languages above,  is again the class  of stochastic languages \cite{YS11A}. It will therefore be sufficient to show how to construct a 1qcfa algorithm, say $M$, that accepts all and only the  members of $\powereqL$ are accepted with probability bigger than $ \frac{1}{2} $, for any given $\LL$. 

Being a  one-way machine, $M$ reads the input from left to the right symbol by symbol in a single pass.
	
	While reading the input, $M$ executes some classical and quantum procedures in parallel. There are two classical procedures:
	\begin{itemize}
		\item If the input is $ aba^7 $, accept. 
		\item Check the input to see whether it is of the form
			\begin{equation}
				\label{eq:well-formed}
				a b a^7 b a^{7 \cdot t_1} b a^{7 \cdot t_2}  \cdots b a^{7 \cdot t_n} 
			\end{equation}
			for some $n> 0 $, where all the $t_i$ are positive multiples of 8. If not, reject.
	\end{itemize}
	
	We now describe the quantum procedures, for which $M$ uses four qubits:
	\begin{enumerate}
		\item Use the first qubit to compare $ t_1 $ with $ \frac{t_2}{8} $, $ t_3 $ with $ \frac{t_4}{8} $, and so on. To compare $ t_{2j-1} $ with $ \frac{t_{2j}}{8} $, we set the qubit to $ \ket{q_0} $, and then rotate it with angle $ \sqrt{2}\pi $ $ t_{2j-1} $  times,  and with angle $ -\sqrt{2}\pi $  $ \frac{t_{2j}}{8} $ times. Then measure the qubit, rejecting if $ q_1 $ is observed. Otherwise, continue with the next pair.
		\item Use the second qubit to apply the same pairwise comparison  to $ t_2 $ and  $ \frac{t_3}{8} $, $ t_4 $ and $ \frac{t_5}{8} $, and so on.
		\item Use the third qubit to flip a fair coin for each $ a $ that is scanned. If any flip results in a ``tail'',  do not take a decision in this procedure.  
Use the fourth qubit to implement the procedure given in Figure \ref{fig:powereqLalg} as long as you keep getting all ``head''s on the third qubit: That is, set the fourth qubit  to $ \ket{q_0} $ at the beginning, and then rotate it with angle $ \theta_L $ (described in Section \ref{sec:rec}) for each $a$ you see. At the end, rotate the qubit with angle $  \frac{\pi}{4}$ and then make a measurement. If the outcome is $ q_1 $, accept. Otherwise,  reject. 
	\end{enumerate}
	If any one of these procedures reaches the right end-marker without a decision, it flips a fair coin to accept and reject with equal probability. 
	
	If the input is not of the form (\ref{eq:well-formed}), then it is accepted with probability 0.	We assume that the input is of the form (\ref{eq:well-formed}) for the rest of the discussion. The analysis of the first and second quantum procedures was given in the proof of Theorem \ref{theorem:powereq}. Therefore, any input that is not in $ \powereq $ is rejected with a polynomially small probability by those procedures. Note that the third quantum procedure will take a decision with exponentially small probability. A majority of the computational paths of the machine will end up contributing  equal amounts to the overall acceptance and rejection probabilities.

	Therefore, for any input not in $ \powereq $, the rejection caused by the first and/or second quantum procedures will overwhelm any acceptance due to the third one, with the input being accepted with total probability less than $ \frac{1}{2} $. If the input is in $ \powereq $, then whether the overall acceptance probability is greater or less than $ \frac{1}{2} $ depends on the tiny but mostly correct contribution from the third procedure. Thus,  $ \powereqL $ is a stochastic language.
\end{proof}



It is still open whether 2qcfa's can recognize any unary non-regular language with bounded error.

We define two more languages: $ \mathtt{UPOWER} = \{a^{8^n} \mid n \geq 0 \} $ and $ \mathtt{UPOWER(L)} = \{a^{8^n} \mid n \in \mathtt{L}\} $. We know that $ \mathtt{UPOWER} $ can be recognized by a two-way deterministic one-counter automaton (2dca) in $ O(|w|\log |w|) $ time, where $ w $ is the given input. A 2qcfa augmented with a classical counter is abbreviated as 2qcca.

\begin{theorem}
	For any given $ \mathtt{L} $, $ \mathtt{UPOWER(L)} $ can be recognized by a 2qcca with bounded-error in $ O(|w|\log |w|) $ time.
\end{theorem} 
\begin{proof}
	Let $ w=a^m $ be the input, where $ m \geq 0 $.	The 2qcca algorithm consists of two phases. It first classically determines whether $ w \in \mathtt{UPOWER}$, by checking if $m$ is a power of 8. If so, it executes the quantum phase, which checks for membership in  $ \mathtt{UPOWER(L)} $ by taking the membership function of $\LL$ into account.

	 In its first  pass over the input $w$, the classical phase rejects if $ m=0 $, or if $m$ is not a multiple of 8, and sets the value of the counter to $  \frac{m}{8} $ otherwise.
	 
	  At this point with the head on the right end-marker, the machine controls the value in the counter, and switches to the quantum phase (described below) if it equals 1. Otherwise, the head is moved exactly $ \frac{m}{8} $ symbols from right to left, setting the counter to 0 at the end of this walk. The head then scans the same $ \frac{m}{8} $-symbol postfix of the input from left to right,  incrementing the counter once every 8 steps,   setting its value to $ \left\lfloor \frac{m}{8^2} \right\rfloor  $ when the right end-marker is reached, and  checking whether the length of this postfix is itself is a multiple 8. If not, the input is rejected, otherwise, the procedure described in this paragraph is repeated. It is clear that the number of zigzags on the input is bounded by $ O(\log |w|) $, and so this procedure takes at most $ O(|w|\log |w|) $ steps.
	
	In the quantum phase, the head is placed on the left end-marker,  and the procedure of Figure \ref{fig:powereqLalg} is implemented on a single qubit. It has already been shown that this procedure accepts with high probability if the number of $a$'s that it scans corresponds to a string in $\LL$, and rejects with high probability otherwise. This procedure takes $ m $ steps.
\end{proof}

Moreover, $\powereq$ can be recognized by 2qcca's in linear time exactly (deterministically). The numbers $i$ and $j$ in each  block  of the form $ a^iba^j $ can be compared while going from left to right, and the overall running time is linear.
\begin{corollary}
	For any given $ \mathtt{L} $, $ \powereqL  $ can be recognized by 2qcca's with bounded error in linear time.
\end{corollary}

We can obtain similar results for other computational models. For example, a one-way two-head deterministic finite automaton can easily recognize $\powereq$ in linear time. Therefore, $ \powereqL  $, for any given $ \mathtt{L} $, can be recognized by a 1qcfa with two input heads in linear time.

Moreover, $ \powereq $ can also be recognized by a 1.5-way\footnote{A 1.5qfa makes a single left-to-right pass on the input. The head is quantum and is allowed to pause on the same square for some steps, so it is possible to have superpositions of several different head positions.} quantum finite automaton (1.5qfa) \cite{AI99} with bounded error in linear time, by a straightforward modification of techniques given in \cite{KW97,AI99,YS09B,Yak12B}. The procedure given in Figure \ref{fig:powereqLalg} can be implemented by such a 1.5qfa in parallel. Therefore, $ \powereqL $, for any given $ \mathtt{L} $, can also be recognized by 1.5qfa's with bounded error.

\end{document}